\documentclass{article}
\usepackage{natbib}

\usepackage{authblk}
\usepackage[utf8]{inputenc}
\usepackage{amsmath,amssymb}
\usepackage{bbm}

\usepackage{subcaption}

\usepackage[ruled]{algorithm2e}
 
\usepackage{hyperref}
\usepackage{cleveref}

\usepackage{amsthm}
\newtheorem{theorem}{Theorem}[section]
\newtheorem{lemma}[theorem]{Lemma}

\newtheorem{definition}[theorem]{Definition}
\newtheorem{corollary}[theorem]{Corollary}

\newtheorem{observation}[theorem]{Observation}

\usepackage{fullpage}

\usepackage[pdftex]{graphicx}

\usepackage{comment}
\usepackage{mathtools}

\usepackage[shortlabels]{enumitem}

\usepackage[shortlabels]{enumitem}

\newcommand{\ignore}[1]{}

\newcommand{\hide}[1]{}

\usepackage[shortlabels]{enumitem}

\usepackage[shortlabels]{enumitem}

\newcommand{\A}{\mathcal{A}}

\usepackage{caption}
\usepackage{xcolor}

\let\Pr\undefined

\DeclareMathOperator*{\Pr}{\mathbb{P}}

\DeclarePairedDelimiter{\ceil}{\lceil}{\rceil}
\DeclarePairedDelimiter{\floor}{\lfloor}{\rfloor}

\newcommand{\cA}{\mathcal{A}}
\newcommand{\cB}{\mathcal{B}}

\newcommand{\cP}{\mathcal{P}}

\newcommand{\esh}[1]{\textcolor{red}{[Eshwar: #1]}}

\newcommand{\ar}[1]{\textcolor{blue}{[AR: #1]}}

\newcommand{\nat}[1]{\textcolor{blue}{[Natalie: #1]}}

\newcommand{\pmax}{p_{\text{threshold}}}

\usepackage{nicematrix}

\usepackage{fullpage}
\title{Algorithmic Collusion Without Threats}
\author[1]{Eshwar Ram Arunachaleswaran}
\author[1]{Natalie Collina}
\author[1]{Sampath Kannan}
\author[1]{Aaron Roth}
\author[2]{Juba Ziani}

\affil[1]{University of Pennsylvania, Department of Computer and Information Sciences}
\affil[2]{Georgia Tech ISyE}

\begin{document}

\maketitle
\begin{abstract} 
There has been substantial recent concern that automated pricing algorithms might learn to ``collude.''  Supra-competitive prices can emerge as a Nash equilibrium of repeated pricing games, in which sellers play strategies which threaten to punish their competitors if they ever ``defect'' from a set of supra-competitive prices, and these strategies can be automatically learned. But threats are anti-competitive on their face. In fact, a standard economic intuition is that supra-competitive prices  emerge from \emph{either} the use of threats, \emph{or} a failure of one party to correctly optimize their payoff. Is this intuition correct? Would explicitly preventing threats in algorithmic decision-making prevent supra-competitive prices when sellers are optimizing for their own revenue?  

No. We show that supra-competitive prices can robustly emerge even when both players are using algorithms which do not explicitly encode threats, and which optimize for their own revenue. Since deploying an algorithm is a form of commitment, we study sequential Bertrand pricing games (and a continuous variant) in which a first mover deploys an algorithm and then a second mover optimizes within the resulting environment. We show that if the first mover deploys \emph{any} algorithm with a no-regret guarantee, and then the second mover even approximately optimizes within this now static environment, monopoly-like prices arise. The result holds for any no-regret learning algorithm deployed by the first mover and for \emph{any} pricing policy of the second mover that obtains them profit at least as high as a random pricing would --- and hence the result applies even when the second mover is optimizing only within a space of non-responsive pricing distributions which are incapable of encoding threats. In fact, there exists a set of strategies, neither of which explicitly encode threats that form a Nash equilibrium of the simultaneous pricing game in algorithm space, and lead to near monopoly prices. This suggests that the definition of ``algorithmic collusion'' may need to be expanded to include strategies without explicitly encoded threats.
\end{abstract}

 \thispagestyle{empty} \setcounter{page}{0}
 \clearpage


\section{Introduction}Consider a market in which there are two sellers of some commodity. We hope and expect that competition should drive the price of the good in the market down to its production cost---for if seller A was consistently selling above cost, then there would be an opportunity for seller B to slightly undercut seller A, capturing the whole market and increasing their profit. In fact, this is the Nash (and correlated) equilibrium outcome in several
 classical formalizations of such competition including the Bertrand Duopoly model. We refer to pricings that are at (or near) the Nash equilibrium prices as \emph{competitive prices}.

As the name suggests, competition is essential to maintaining competitive prices: a monopolist could freely set the price of a good as high as the market could support (which, for goods with inelastic demand, could be quite high indeed) without fear of being undercut by a competitor. So, if in a market with multiple participants, we see prices that are closer to the monopoly price than the competitive price, we might suspect ``anti-competitive'' behavior on the part of the sellers. 

Yet, defining anti-competitive behavior turns out to be a difficult exercise. The Sherman act has been interpreted to require express agreement to coordinate on prices in the form of overt communication as a pre-requisite for liability. However, a recent line of work has focused on the ability of \emph{pricing algorithms} to arrive at supra-competitive prices without explicit communication---concerns rooted in real world observations \citep{assad24} and a growing academic literature \citep{ABADA2024,lamba2022,hansen2021frontiers}.  The mechanism underlying most results in this literature ultimately are rooted in ``folk-theorem'' style equilibria of repeated pricing games~\citep{benoit1984finitely,littman2003polynomial}. Broadly speaking, folk theorems establish that all pairs of prices that guarantee the sellers at least the profit they obtain with competitive pricing can be realized as equilibria of the repeated game. Crucially, the strategies which are shown to lead to these anti-competitive equilibria have explicitly encoded threats aimed at punishing the other seller if they deviate from a proscribed pricing path. Further, reinforcement learning algorithms can automatically discover these strategies; recent work~\citep{calvano20,klein2021} has shown both empirically and analytically that these algorithms can converge to supra-competitive prices, and they do so by the emergence of punishments and threats, even without being explicitly programmed to learn such strategies. Using such algorithms may not violate current anti-trust law as they do not involve explicit communication, but one can imagine future legislation and case law that forbids the use of algorithms that make use of threats (appropriately defined). An algorithm is after all an explicit commitment to action as a function of its observations, and committing to an algorithm might be interpreted as communicating the content of the threat encoded within the algorithm. 

To formally rule out ``anti-competitive'' behavior, an approach one might take is to define algorithms that do not encode threats.  In fact, the standard economic stance \citep{harrington,calvano20} has been that threats are an \emph{explicit} requirement for collusion, distinguished from other situations in which supra-competitive prices can arise due to a ``failure to optimize:''
\begin{quote}
    ``To us economists, collusion is not simply a synonym of high prices but crucially involves ``a reward-punishment scheme designed to provide the incentives for firms to consistently price above the competitive level'' (Harrington 2018, p. 336). The reward-punishment scheme ensures that the supra-competitive outcomes may be obtained in equilibrium and do not result from a failure to optimize.'' \citep{calvano20}
\end{quote}
 A first attempt might be to say that (at least if there are other optimizers in the market) ``non-responsive'' algorithms which set prices completely independently of competitor prices (and consequences of competitor prices) should be permitted---indeed, if they cannot react to their competitors, then they cannot threaten them. But this is far too limiting as it does not allow for algorithms that react to market conditions. Recent work by~\citet*{hartline2024regulation} and~\citet*{chassang23} set out a principled definition of what constitutes non-collusive behavior of a pricing algorithm by defining algorithms that satisfy some internal consistency properties (calibrated regret, also known as ``swap regret'') and then shows that if all algorithms in the market satisfy these conditions then the outcome will be competitive prices. Can this condition be relaxed? What if the first entrant to a market deploys (and hence commits to) an algorithm that satisfies these internal consistency conditions, and then the second entrant plays a simple policy that does not satisfy the swap regret condition, but also does not encode threats (e.g. because it is entirely non responsive)?  Might we still see competitive prices? Would we, in accordance with the intuition of~\citep{harrington, calvano20} see competitive prices if the follower not only did not deploy threats, but also succeeded in ``optimizing''? 

In this paper we show that the answer is robustly ``no'' by showing how a broad class of ``reasonable'' behavior on the part of all parties must lead to supra-competitive (near monopoly) prices. In particular, with the view that deploying an algorithm is inherently a ``commitment'' (at least in the short term), we study two seller pricing games in a leader-follower model, in which the leader first deploys an algorithm, and then the follower deploys an algorithm which may optimize within the market environment that is implicitly defined by the leader's pricing algorithm. If both sellers viewed this scenario as a Stackelberg game, and optimized over the set of all pricing algorithms,  it would hardly be surprising if anti-competitive behavior emerged. Indeed, the Stackelberg equilibria in algorithm space (See~\citep{collina2023efficient}) in a repeated pricing game would involve the leader committing to an algorithm that prescribes monopoly prices and commits to punish the follower if she deviates from this prescription. 

However, we instead study the behavior that emerges if sellers attempt to play reasonable algorithms. In particular, we assume the leader (following the principle set out by \citep{hartline2024regulation,chassang23}) commits to setting prices using an arbitrary \emph{no-regret} algorithm. No-Regret algorithms are widely understood to be a ``natural" class of algorithms for use in competitive environments~\citep{nekipelov2015econometrics}, and our results will hold for any algorithm in this class, which includes the no-swap regret algorithms studied by \citep{hartline2024regulation,chassang23}. 

Once the leader deploys a no-regret learning algorithm, the other seller (whom we call the follower, or the optimizer) deploys an algorithm to optimize their own revenue in the environment defined by the leader's no-regret learning algorithm. Once again, we do not assume that the follower is best responding in algorithm space: they may, e.g., use a heuristic reinforcement learning algorithm to optimize over a limited set of policies---even \emph{non-responsive} policies that are explicitly independent of the leader's prices, and hence cannot encode threats.  We show that however they choose to optimize their own revenue, if they are even marginally successful (choosing any policy that guarantees them at least the revenue that a random pricing strategy would), then in combination, the two algorithms will inevitably lead to supra-competitive prices in a variety of settings. When this happens, it is not just the follower who enjoys the revenue of the high prices, but as we prove, also the leader (if the leader deployed a no-regret learning algorithm). Thus the leader will  have no strong incentive to deviate from their deployed algorithm, even if they deployed it without initially being aware of the competitive environment. In fact, within the class of ``reasonable'' strategies we study for the leader and follower (a no-swap-regret algorithm for the leader and a non-responsive algorithm for the follower), there exists a Nash equilibrium of the simultaneous pricing game in algorithm space, which gives \emph{no party} any incentive to deviate, supports near-monopoly prices, and does not involve threats. This suggests that it may be worth revisiting the standard economic viewpoint \citep{calvano20,harrington} that \emph{threats} are a necessary precondition for collusion.

\subsection*{Summary of Results}

We study two types of models: a simple Bertrand duopoly model, where the firm with the lowest price captures the entire demand, and a smoother multinomial-logit-based model where firms with higher prices can still capture some of the demand. Our multinomial logit model includes a temperature parameter $\tau$ that controls how sensitive the demand for each firm is to prices, and  converges to a Bertrand competition when the parameter $\tau$ grows large. Within these models, we prove the following results:

\begin{enumerate}
    \item The second entrant into the market (whom we call the optimizer) can extract a constant fraction of the monopoly price revenue against any no-regret algorithm by playing a non-responsive strategy (i.e. a fixed distribution of prices across all rounds, that does not depend on or respond to the history of play). As a result, if they approximately optimize their pricing policy (in any set of policies that includes natural fixed price distributions), they will be guaranteed a constant fraction of the monopoly revenue. Observe that this holds even if the class they optimize over \emph{only} includes non-responsive pricing strategies, which are incapable of encoding threats.
    \item Simultaneously, by deploying any no-regret algorithm, the first entrant (whom we call the learner)  extracts a constant fraction of the revenue obtained by any optimizer strategy that obtains a constant fraction of the monopoly price revenue. Taken in conjunction with the previous result, this implies that if the first entrant into the market deploys (and hence commits to) a no-regret learning algorithm, and the optimizer deploys any strategy that approximately optimizes over any class of policies that includes fixed price distributions, \emph{both} parties will extract a constant fraction of the monopoly price revenue. This has important implications: not only do supra-competitive prices emerge, but the benefits are shared across both sellers, giving neither a substantial incentive to deviate. In other words, the optimizer is not ``exploiting'' the no-regret learner, but rather implicitly cooperating with them, despite the absence of threats.
    \item In fact, we show that if the first entrant into the market deploys a no-swap regret algorithm, then there is a non-responsive (i.e. fixed) price distribution for the optimizer that forms a \emph{Nash} equilibrium of the simultaneous pricing game in algorithm space: i.e. \emph{neither} player has any incentive to deviate from their selected algorithm. This demonstrates the existence of a Nash equilibrium of the game that supports supra-competitive prices and yet does not involve threats, contra to the economic intuition \citep{calvano20,harrington} that supra-competitive prices result from \emph{either} threats or a failure to optimize. 
\end{enumerate}

When we refer to supra-competitive prices in the above theoretical results, we mean prices that are within a constant factor of the monopoly price. The constants in our theorems are small, but we provide numerical evidence that the constant is in fact no smaller than $2/e \approx 0.74$. 

We also prove an additional result about the Bertrand duopoly model:

\begin{itemize}
    \item Our results apply to \emph{any} no regret learning algorithm deployed by the first entrant into the market. This is a super-set of the no-swap regret learning algorithms proposed to be definitionally ``competitive'' by \citep{hartline2024regulation,chassang23}. Should no regret algorithms beyond no-swap regret algorithms also be considered competitive? This is not the case for all no-regret dynamics---~\cite{nadav2010no} show the existence of coarse correlated equilibria supported on supra-competitive prices in Bertrand pricing games, implying the existence of no-regret dynamics that result in such prices. 
    However, we show that when both sellers use Mean-based No-Regret (MBNR) algorithms (a class which includes common no-regret learning algorithms such as multiplicative weights and follow the perturbed leader) to set prices, the empirical distribution of play always converges to the competitive price. 
    This suggests that at least in the Bertrand Duopoly model, mean based no-regret learning algorithms might also be considered reasonable/competitive pricing strategies, in addition to the no-swap regret learners proposed by \citep{hartline2024regulation,chassang23}.  We note that this result does not necessarily carry over to the more complex pricing models studied by \citep{hartline2024regulation}.
    
\end{itemize}



\subsection{Related Work}

The seminal work by \cite{calvano20} experimentally shows that Q-learning algorithms, a simple and central type of reinforcement learning algorithm, can learn to collude with each other via threats in a repeated, simultaneous-action pricing game. These algorithms learn to deviate to lower prices in response to another algorithm doing so, and slowly rise back to higher prices. \cite{klein2021} attain similar results in the sequential repeated pricing game. The authors in these works are specifically concerned with the ability of the algorithms to learn clear threats, as opposed to simply reaching supra-competitive prices. This is  because, from the perspective of many economists, such threats are a requirement for collusion. The prevailing view, as espoused by \cite{harrington} and \cite{calvano20}, is that any algorithms which reach supra-competitive prices without threats must be suboptimal, and thus uninteresting as candidate pricing algorithms. 

This perspective is supported by \cite{Abada2020} (experimentally) and \cite{ABADA2024} (theoretically). These works consider a different class of reinforcement learning algorithms and show that, while they do not converge to supra-competitive prices when they have full computational capacity, they can reach supra-competitive prices if these algorithms are limited in their ability to explore. These supra-competitive prices arise even though these algorithms do not retain enough state information to encode threats. Therefore, while these works do show supra-competitive prices without threats, they are only shown arising from algorithms which are specifically designed to be suboptimal. 

Another line of research has explored the emergence of supra-competitive prices as a byproduct of algorithmic symmetry. \cite{hansen2021frontiers} show that, if two identical, deterministic algorithms are played against each other, their ``exploration'' will be symmetric, and hence they will be unable to ever observe the benefits of undercutting their competitors (an inherently asymmetric outcome). 

~\cite{banchio2023} study algorithms that attempt to estimate the payoff of actions through stochastic experiments. They show that in certain settings in which players update the reward for the action they play, but not other actions, then the experiments can correlate the actions of the players, leading to supra-competitive prices. In contrast, the algorithms we study operate in the full information setting (because agents can observe the price played by the other agent, they can estimate the profit they would have made had they played a different price at a given time period), and are designed for adversarial settings, which are well specified in competitive environments: they are designed with the understanding that there are no reward distributions to estimate because profits depend on the evolving actions of an opponent.

In contrast to these previous works, which argue variously that the property of algorithms which leads to supra-competitive pricing are threats and sub-optimality, we argue that the core is in the fact that deploying an algorithm is a form of commitment. This is most in line with the perspective of \cite{lamba2022}. This work treats algorithms as a form of commitment which contributes to collusion. However, the algorithms explored in that work are constrained to simple, two-state automata, which severely limits the strategy space. By contrast, we consider the space of all possible algorithms mapping histories to future prices, and show results for broad families of ``reasonable'' algorithms --- not just equilibrium strategies.

There is also a recent line of work on regulating pricing algorithms. Both \cite{chassang23} and \cite{hartline2024regulation} suggest that observed \emph{swap regret} might be taken as a proxy for collusive behavior. They show, in a related but more general model than our own, that if both sellers in a duopoly have vanishing swap regret, the time-average price must converge to the competitive price. Our work highlights the importance of \emph{both} sellers having no-swap regret; as we show, if the first entrant into a market deploys a no swap regret algorithm, then many natural behaviors for the second entrant lead to supra-competitive prices. 

Finally, there is a recent line of experimental work demonstrating anti-competitive behavior in large language models, including market division in the Cournot competition model~\cite{lin2024} and, more similarly to this work, price fixing in the Bertrand competition model~\cite{fish2024}. These works underscore the importance of understanding the theoretical underpinnings of algorithmic collusion. 

The emergence of anti-competitive behavior in no-regret dynamics is also explored in~\cite{Kolumbus2022}, though in a somewhat different setting---this work considers agents deploying no-regret algorithms to bid for them in an auction, and shows that agents misreporting their true item value to the algorithms can lead to collusive outcomes.


\section{Model}

\subsection{The Single-Stage Pricing Game} 

We will first formally define the  payoff model of the stage game, in which both sellers select a price only once and specify the payoff they obtain.

We work in a duopoly model, where two sellers are competing to sell to a single buyer. Both sellers pick from a discrete set of prices $\cP = \{\frac{1}{k}, \frac{2}{k},\cdots 1\}$. Here, $k$ is a discretization parameter, and we assume that $k \ge 20$.

\paragraph{Game definition} We consider two sellers, who we will refer to as the Sellers 1 and 2. 
The action space of each seller is the set of possible prices in $(0,1]$, up to some discretization parameter $1/k$. Both sellers can choose to distributions of prices from the set $\cP$, and we denote the set of all distributions by $\Delta^{\cP}$. 

To define our game, we first need to define the concept of an allocation rule. An \emph{allocation rule} is a (potentially randomized) rule that divides demand across the two sellers as a function of their chosen prices:

\begin{definition}[Allocation Rule $C_{i}(p_{1},p_{2})$]
An \emph{allocation rule} $C_{i}: \cP \times \cP \mapsto [0,1]$ is a mapping from a price pair $(p_{1},p_{2})$ from  the two sellers, respectively, to a fraction of the demand that will be allocated to seller $i \in \{1,2\}$. 
\end{definition}

Given an allocation rule that defines and controls market outcomes as a function of prices, we can now define the payoff of a seller in our single-stage game as a function of said  allocation rule $C$. 

\begin{definition}[Seller Payoff in the Stage Game]\label{def:utilities}
Given an allocation rule $C(.,.)$, the payoff of seller $i$ is given by
\begin{equation}
u_{i}(p_{1},p_{2})= p_{i} \cdot C_{i}(p_{1},p_{2}).
\end{equation}

When the sellers play distributions, the payoff of the seller is understood to be the expected payoff.
\end{definition}

We define the buyer price as the price charged by the seller that is the output of the allocation rule.
\begin{definition}[Average Buyer Price in the Stage Game]
\label{def:bprice}
For an allocation rule $C$ and (mixed) strategies $p_1$, $p_2$, the average price paid by the buyer is $\mathbbm{E}[p_1 C_1(p_1,p_2) + p_2 C_2(p_1,p_2)]$. Note that this means that the average buyer price is equal to the (expected) sum of payoffs of the two sellers.
\end{definition}

Throughout this paper we assume that the \emph{total} demand (added across both sellers) for the item being sold is fixed and independent of the prices. Without loss of generality, this total demand is set to be normalized to $1$. As a result, the ``monopoly price'' in these models is always $1$. Within this fixed-demand model, we consider two allocation rules: a Bertrand model and a multinomial-logit-based model, both formalized in Section~\ref{sec:competition_models}. Both allocation rules are {\it symmetric} and therefore the resulting games are {\it symmetric } as well.


We can now formally define the one-stage game studied in this paper, often informally referred to simply as the ``stage game'':
\begin{definition}[Pricing Stage Game]
    A two-player pricing stage game $G(k,C)$ is defined as a combination of: 
    \begin{enumerate}
    \item Two sellers, respectively named $1,~2$;
    \item An action space $A = \cP = \{\frac{1}{k},\frac{2}{k},\ldots,1\}$ corresponding to possible prices that can be played by $1$ and $2$. We denote by $p_1$ the action of Seller $1$ and $p_2$ the action of seller 2. These actions may be picked from the set of distributions over $\cP$, which we call $\Delta^{\cP}$.
    \item Payoff functions $u_i(p_1,p_2)$ for $i \in \{1,2\}$, defined as a function of the competitive allocation rule $C$ as in  Definition~\ref{def:utilities}.
    \end{enumerate}
    A particular pricing game is defined by the allocation rule $C$ and the discretization parameter $k$. We will work with symmetric allocation rules (the allocation does not depend on the labels of the sellers), which will naturally induce a symmetric finite game between the two sellers.
\end{definition}

The one-stage game described is this section is the starting point of our study, one that helps us define competitive prices as the Nash equilibria outcome of the game. However, the primary object of our study is \emph{repeated} interactions between sellers, described in Section~\ref{sec:repeated_game}.

\paragraph{Equilibria of the Stage Game}
We will refer to two kinds of equilibria of the stage game which differ in their assumed timing. The \emph{Nash} equilibria characterize outcomes in simultaneous play, whereas \emph{Stackelberg} equilibria characterize outcomes of optimal play when one player may first commit to a strategy, and the other player then responds.

\begin{definition}[Nash Equilibrium of the Stage Game]
A pair of (independent) distributions $(p_1,p_2)$ for the two sellers  is said to be an $\varepsilon$-Nash equilibrium  if $u_1(p_1,p_2) \ge \max_{p \in \cP} u_l(p,p_2) - \varepsilon$ and $u_2(p_1,p_2) \ge \max_{p \in \cP} u_2(p_1,p) - \varepsilon$. When $\varepsilon = 0$, this pair is a Nash equilibrium.
\end{definition}

\begin{definition}[Stackelberg Equilibrium of the Stage Game]
A pair of (independent) distributions $p_1,p_2$ for the two sellers is said to be an $\varepsilon$-Stackelberg equilibrium of the stage game with seller 1 as the leader if 
\[
u_2(p_1,p_2) \ge \max_{p \in \cP} u_2(p_1,p),
\]
and 
\begin{align*}
u_{1}(p_{1},p_{2}) \geq 
\max_{(q_{1}, q_{2} )\in \Delta^{\cP}} &u_{1}(q_{1},q_{2}) - \varepsilon\\
\text{s.t.}~~& u_2(q_1,q_2) \ge \max_{q \in \cP} u_2(q_1,q).
\end{align*}
When $\varepsilon = 0$, this pair is a Stackelberg equilibrium.
\end{definition}

\subsubsection{Competition Models}\label{sec:competition_models}
Now that we have defined the pricing stage game in full generality, we will instantiate it with two different specific competition rules, which result in two well-known economic pricing models. 

\paragraph{Bertrand Competition: } In the Bertrand Competition model, the seller picking the lower price captures all of the demand. The seller picking the higher price earns nothing. For cases in which sellers chose the same price $p_L$, they split the market equally.

\begin{definition}[Bertrand Competitive Allocation Rule $C^{B}_{i}$] $C^{B}_{i}$ is a competition rule such that, given player $j \neq i$,

\begin{equation}
C^{B}_{i}(p_{1},p_{2}) =
\begin{cases}
1 & p_{i} < p_{j},\\
\frac{1}{2} & p_{i} = p_{j},\\
0 & p_{i} > p_{j}
\end{cases}    
\end{equation}

\end{definition}

\begin{definition}[Bertrand Pricing Stage Game $G^{B}(k)$]
    The \emph{Bertrand Pricing Game} is a pricing game instantiated with the competition rule $C^{B}$ and with a pricing discretization of $\frac{1}{k}$.
\end{definition}


The Bertrand model is discontinuous in that the seller who selects the lower price captures the entire market. We also study a generalization of Bertrand competition in which the fraction of the market captured by a seller is a smooth function (a logit function) of the difference in prices set by the two sellers. This model has a temperature parameter $\tau$, and as $\tau$ tends to infinity, it recovers the Bertrand competition model as a special case.


\paragraph{Multinomial-logit-based Price Competition: } 
This competition model is parameterized by a temperature parameter $\tau \geq 0$ 

\begin{definition}[Logit Competitive Allocation Rule $C^{L, \tau}_{i}$] $C^{L,\tau}_{i}$ is a competition rule such that for $j \neq i$:
\begin{equation}
C^{L,\tau}_{i}(p_{1},p_{2}) =
\frac{e^{\tau p_{j}}}{e^{\tau p_{i}}+e^{\tau p_{j}}}
\end{equation}
\end{definition}

\begin{definition}[Multinomial-logit-based Pricing Stage Game $G^{B}(k,\tau)$]
    The \emph{Multinomial-logit-based Pricing Game} is a pricing game instantiated with the competition rule $C^{L,\tau}$ and with a pricing discretization of $\frac{1}{k}$.
\end{definition}


We prove similar results as in the Bertrand model in this logit model, assuming that the temperature parameter is sufficiently large, i.e., the allocation rule is sufficiently sensitive to differences in price. As $\tau \rightarrow \infty$, we recover the Bertrand model in the limit, since the lower price seller wins the demand with probability $1$. Likewise, when $\tau = 0$, the sellers each win with probability 1/2 regardless of their price, so it is to be expected that we need $\tau$ to be reasonably large to prove any meaningful results.

This model is a special case of the logit model of~\cite{calvano20} and is obtained by setting all ``product indices'' in their model to be equal, and by removing the outside good from consideration (by setting  its product index $a_0 $ to $ -\infty$), the latter action fixing the total demand.

For simplicity, in the main body of the paper we prove our results in the Bertrand game (Section~\ref{sec:bertrand}). We show that they generalize to the Multinomial-Logit-Based pricing game in Appendix~\ref{app:logit}.

\subsection{The Repeated Pricing Game}\label{sec:repeated_game} 

We now move from the stage game to the repeated pricing interaction that will be our main setting of interest. 
The  pricing game is played repeatedly over $T$ rounds. During each round $t$, the two sellers simultaneously 
pick mixed strategies $p_{1,t},~p_{2,t}$ over the set of prices $\cP$ (these are independent distributions over $\cP$). The sellers receive the corresponding utilities equal to the expected payoff of pricing $p_{1,t}$ against $p_{2,t}$ in the stage game. Both sellers observe the full mixed strategy of their opponent.
Here, the pricing behavior of each player will be defined by algorithms which map history to the pricing distribution at the next round, and we study this larger interaction as a game in which the strategy space for each player is the algorithm that they will deploy.  We formalize this setting below by defining the strategy space and the utility functions.

\paragraph{Algorithms as Strategies}
The sellers' strategy space is the space of all algorithms, mapping histories to price distributions at the following round in arbitrary ways.

\begin{definition}
    An algorithm $\cA$ for the $T$ round repeated game is  
    a mapping from the history of play (without loss of generality of only the other seller) to the next round's action, denoted by a  collection of $T$ functions $A_1, A_2 \cdots A_T$, each of which deterministically map the transcript of play (up to the corresponding round) to a mixed strategy to be used in the next round, i.e., $A_{1,t}(p_{2,1}, p_{2,2},\cdots, p_{2,t-1}) = p_{1,t}$ for the algorithm $\cA_1$ used by seller 1.
\end{definition}

We define a restricted subclass of algorithms below---algorithms that do not use nor respond to the history of play. Such algorithms, in particular,  cannot make use of or respond to threats as they are oblivious to the behavior of their opponent.
\begin{definition}
    If $A_t$ is independent of its input for all $t$, we call it \emph{non-responsive}. If in addition $\cA_t$ outputs the \emph{same} distribution for all $t$ we call it static.
\end{definition}
 

If the two sellers pick algorithms $\cA_1$ and $\cA_2$ respectively, this induces a transcript of $T$ price distribution pairs $((p_{1,1},p_{2,1}), \cdots (p_{1,t},p_{2,t}) \cdots (p_{1,T},p_{2,T}))$
 where $A_{1,t}(p_{2,1}, p_{2,2},\cdots, p_{2,t-1}) = p_{1,t}$ and, similarly, $A_{2,t}(p_{1,1}, p_{1,2},\cdots, p_{1,t-1}) = p_{2,t}$ for each round $t$.

\paragraph{Game definition} We again consider two sellers, who we will now refer to as the \emph{learner} and the \emph{optimizer}. The action space of each seller is now the set of possible \emph{algorithms} mapping the history of play to a pricing distribution at the current round.

\begin{definition}[Seller Payoff in the Repeated Game]\label{def:rep_utilities}
Let $p_{l}^{1:T}$ and $p_{o}^{1:T}$ represent the sequence of $T$ distributions over prices by the learner and the optimizer, respectively, given $\cA_{l}$ and $\cA_{o}$. Then, the payoff of seller $i \in \{l,o\}$ is
\begin{equation}
U_{i}(\cA_{l},\cA_{o})= \sum_{t=1}^{T} u_i(p_{l}^{t},p_{o}^{t})
\end{equation}
\end{definition}

We can now formally define the central game studied in this paper:
\begin{definition}[Repeated Pricing Game]
    A two-player repeated pricing game $G(k,C,T)$ is defined as a combination of: 
    \begin{enumerate}
    \item Two players, named the ``learner'' and the ``optimizer'';
    \item A time horizon $T$ specifying the length of the repeated game;
    \item An action space $A = \mathcal{A}^{T}$ corresponding to possible sequential pricing algorithms for games of length $T$; 
    \item Payoff functions $U_l(\cA_{l},\cA_{o})$ for the learner and $U_o(\cA_{l},\cA_{o})$ for the optimizer, defined as in Definition~\ref{def:rep_utilities}.
    \end{enumerate}
    A particular repeated pricing game is defined by the allocation rule $C$ and the discretization parameter $k$. 
\end{definition}


\paragraph{Equilibrium Notions in Algorithm Space}
Once again, we can study two kinds of equilibria of the repeated game in ``algorithm space'', depending on the timing of the game: Nash equilibria if play is simultaneous, and Stackelberg equilibria if the learner has commitment power and moves first. 
\begin{definition}[Nash Equilibrium in Algorithm Space]
A pair $(\cA_l,\cA_o)$ of algorithms for the two sellers is said to be an $\varepsilon$-Nash equilibrium in ``algorithm space''  if $U_l(\cA_l,\cA_o) \ge \max_{\cA \in \cA^T} U_l(\cA,\cA_o) - \varepsilon$ and $U_o(\cA_l,\cA_o) \ge \max_{\cA \in \cA^T} U_o(\cA_l,\cA) - \varepsilon$. When $\varepsilon = 0$, we call this pair a Nash equilibrium in algorithm space.
\end{definition}


\begin{definition}[Stackelberg Equilibria in Algorithm Space]
A pair $(\cA_l,\cA_o)$ of algorithms for the two sellers is said to be an $\varepsilon$-approximate Stackelberg equilibrium in ``algorithm space'' if 
\[
U_o(\cA_l,\cA_o) \ge \max_{\cA \in \cA^T} U_o(\cA_{l},\cA),
\]
and 
\begin{align*}
U_{l}(\cA_{l},\cA_{o}) \geq 
\max_{(\cB_{l}, \cB_{o} )\in \A^{T}} &U_{l}(\cB_{l},\cB_{o}) - \varepsilon\\
\text{s.t.}~~&U_o(\cB_l,\cB_o) \ge \max_{\cB \in \cA^T} U_o(\cB_{l},\cB)
\end{align*}
When $\varepsilon = 0$, we call this pair a Stackelberg equilibrium in algorithm space.
\end{definition}

We investigate a setup in which one seller, who we call the Learner, commits to an algorithm $\A_{L}$, and then the other seller, who we call the Optimizer, responds with their own algorithm $\A_{O}$~\footnote{We follow the convention in prior work in renaming the leader and follower as the learner and optimizer where one of the key cases considered is the leader playing a no-(swap-)regret algorithm. This is helpful in clarifying the roles of the players, as we will use the fact that the optimizer, playing against a no-swap-regret algorithm, can get their Stackelberg \emph{leader} value in the stage game}. This is reminiscent of a Stackelberg game in algorithm space; however, in contrast to standard Stackelberg games, we do not necessarily assume that the Optimizer exactly best-responds to the Learner's algorithm, or that the Learner necessarily deploys the algorithm representing the optimal commitment strategy. Instead, we examine which pricing outcomes arise under various Optimizers of different capacities, and show that, when the Optimizer is responding to a no-regret algorithm (defined in Section \ref{sec:regret}), a robust set of responses, including exactly optimal but also including approximately optimal non-responsive strategies, lead to supra-competitive prices.

\paragraph{Competitive vs supra-competitive prices} In the economics literature, competitive prices are defined as prices that can emerge from the static Nash equilibria of the pricing game (i.e. Nash equilibria of the stage game, where prices, rather than algorithms, form the strategy space). We formalize this below:

\begin{definition}[Competitive and Supra-competitive prices] For some pricing model $G$, let $p$ be the highest price that is in the support of any  Nash equilibrium of the stage pricing game. Then, $p'$ is a competitive price if $p' \leq p$. Furthermore, $p'$ is a supra-competitive price if $p' = \Omega_k(1)~\footnote{The $\Omega_k(1)$ notation means that this quantity remains a constant between $0$ and $1$ regardless of the growth of $k$, in contrast to an expression such as $\frac{1}{k}$.}$.  
\end{definition}

Supra-competitive prices are sometimes defined as any prices that are strictly above the competitive price. In this work, we use and guarantee a stronger definition, and require supra-competitive prices to be a constant fraction of the maximum possible price ($1$), even when the competitive price tends to $0$ as $k\rightarrow \infty$. Thus our definition of supra-competitive prices really refers to ``near monopoly'' prices. 

Throughout the paper we will regularly use high seller payoff as proof of supra-competitive prices. We make that connection explicit here, showing that in both the Bertrand and Logit models, the average price over time is exactly equal to the sum of the average utilities of the two sellers:

\begin{definition}[Average Buyer Price in the Repeated pricing Game]
For an allocation rule $C$ and seller algorithms $\cA_l$, $\cA_o$, the average buyer price over the rounds is simply the time average of the buyer prices over each round induced by the seller's prices and the allocation rule in the resulting transcript of play (see Definition~\ref{def:bprice}). 
\end{definition}

We make a simple observation connecting the average payoffs of the sellers with the average prices over the rounds.
\begin{observation}
\label{obs:value_transfer}
    The average buyer price in a repeated pricing game exactly equals the sum of the average payoffs of the sellers.
\end{observation}

\section{Online Learning Preliminaries}
\label{sec:regret}


\paragraph{Classes of Algorithms} We are interested in three specific classes of  algorithms that the learner might commit to: \emph{no-regret algorithms}, \emph{no-swap-regret algorithms}, and \emph{mean-based algorithms}. The learner and optimizer repeatedly play a single-stage game $G$ for $T$ rounds. We refer to the learner's payoff function (in the stage game) as $u_L$ and the actions chosen by the learner and optimizer in round $t$ as $x_t \in \Delta^n$ and $y_t \in \Delta^m$ respectively (the action spaces are distributions over finite sets of actions).
While we describe these algorithms from the perspective of the learner, they are also valid algorithms for the optimizer and afford them the same properties against algorithms picked by the learner.

\paragraph{No-regret definitions} We consider a setting in which a learner faces an adversary, here called an ``optimizer'' that can choose an arbitrary sequence of actions, with knowledge of the learner's algorithm. The learner obtains a payoff $u_L(x_t,y_t)$ as a function of his own action $x_t$ and the optimizer's action $y_t$. A desirable property is that in hindsight, after the sequence of learner and optimizer actions have been realized, the learner does not \emph{regret} not having played a simple strategy (like consistently playing whichever turned out to be the best fixed action in hindsight). We say that an algorithm $\cA$ has \emph{regret} $r(T)$ on some sequence of length $T$ if, regardless of the sequence of optimizer actions, the learner is guaranteed that in hindsight, they obtained payoff at least as high as they could have with the best fixed action, minus $r(T)$. 

\begin{definition}[$r(T)$-Regret Algorithm] 
An algorithm has worst-case regret $r(T)$ if, for any sequence of actions $(y_1, y_2, \dots, y_T)$ taken by the optimizer, the total  payoff of the learner can be lower bounded by
\begin{equation*}
\sum_{t=1}^{T} u_{L}(x_t, y_t) \geq \left(\max_{x^* \in [n]} \sum_{t=1}^{T} u_{L}(x^*, y_t)\right) - r(T).
\end{equation*}
An algorithm is a \textbf{no-regret algorithm} if it is an $r(T)$-regret algorithm with $r(T) = o(T)$.
\end{definition}


We also consider a strict subset of the class of no-regret algorithms, the class of \emph{no-swap regret} algorithms. An algorithm $\cA$ is a no-swap-regret algorithm if it has the no regret property not just marginally, but \emph{conditionally} on each action it played. This can be formalized by requiring that the learner's cumulative  payoff is at least as high as it would have been had they, retrospectively, been able to apply some \emph{swap function} $\pi:[n]\rightarrow [n]$ to their sequence of actions in hindsight. 

\begin{definition}[$r(T)$-Swap Regret Algorithm] 
An algorithm has worst-case swap regret $r(T)$ if, for any sequence of actions $(y_1, y_2, \dots, y_T)$ taken by the optimizer, the total  payoff of the learner can be lower bounded by
\begin{equation*}
\sum_{t=1}^{T} u_{L}(x_t, y_t) \geq \max_{\pi: [n] \rightarrow [n]} \sum_{t=1}^{T} u_{L}(\pi(x_t), y_t) - r(T).
\end{equation*}
where $\pi(x_t)$ refers to the linear extension of $\pi$ acting on the support of $x_t$.

An algorithm is a \textbf{no-swap regret algorithm} if it is an $r(T)$-swap regret algorithm with $r(T) = o(T)$.
\end{definition}

\noindent
Here the maximum is over all swap functions $\pi: [n] \rightarrow [n]$ (extended linearly to act on elements $y_t$ of $\Delta_n$). It is a fundamental result in the theory of online learning that both no-swap-regret algorithms and no-regret algorithms exist (see \cite{cesa2006prediction}).

Some no-regret algorithms have the property that at each round, they approximately best-respond to the historical sequence of losses interpreted as a mixed strategy. Following \cite{braverman2018selling} and \cite{deng2019strategizing}, we call such algorithms \emph{mean-based algorithms}. Formally, we define mean-based algorithms as follows.

\begin{definition}[Mean-Based Algorithm]\label{def:mean-based}
An algorithm $\cA$ is \emph{$\gamma(t)$-mean-based} if whenever $j, j' \in [m]$ satisfy

$$\frac{1}{t}\sum_{s=1}^{t} u_L( j',y_s) - 
\frac{1}{t}\sum_{s=1}^{t} u_L( j,x_s) \geq \gamma(t),$$

\noindent
then $x_{t, j} \leq \gamma(t)$ (i.e., if $j$ is at least $\gamma(t)$ worse than some other action $j'$ against the historical average action of the opponent, then the total probability weight on $j$ must be at most $\gamma(t)$). A learning algorithm is \emph{mean-based} if it is $\gamma(t)$-mean-based for some $\gamma(t) = o(1)$.
\end{definition}

Many standard no-regret learning algorithms are mean-based, including Multiplicative Weights, Hedge, Online Gradient Descent, and others (see \cite{braverman2018selling}).





Up to low order terms, best responses to no-swap regret algorithms in algorithm space are well understood: the optimizer can do no better than to play a static non-responsive strategy which at every round plays an identical action distribution that is very close (in any reasonable norm) to their optimal Stackelberg leader strategy in the stage game:

\begin{lemma}[\cite{deng2019strategizing}]
  \label{lemma:nsr_bestresponse}
  If the learner plays a no-swap-regret algorithm, then there exists $\varepsilon = o_T(1)$ and a $o(T)$-additive best-response of the optimizer that involves playing a distribution $D'$ in each round such that $||D-D'||_{\infty} \le \varepsilon$ where $D$ is the optimizer's leader Stackelberg strategy of the underlying stage game.
\end{lemma}


\section{Results in the Bertrand Model}
\label{sec:bertrand}

\subsection{Equilibria of the Stage Game}

In this section, we focus on Nash and Stackleberg Equilibria of the stage game. We remind the reader that the Nash Equilibria of the stage game define which prices are competitive. In turn, this section identifies what the baseline competitive prices are in the Bertrand model (the equivalent result for the logit model can be found in Appendix~\ref{app:logit}).

We begin by showing that in any Nash equilibrium of the Bertrand stage game, the resulting competitive prices must be at most $\frac{2}{k}$, implying in particular that competitive prices yield in total a vanishing fraction of the monopoly revenue. To prove this, we need to show that pricing high relative to the other seller's price distribution is a dominated strategy, in a robust sense. Intuitively, it is clear that the best-response to a deterministic price is to undercut your opponent. We show that the same idea works against distributions in the following lemma (proof in appendix~\ref{app:marginal_domination}) that we use throughout this section.

\begin{lemma}
\label{lemma:marginal_domination}
Consider one seller (say seller 2) in the Bertrand stage game playing a distribution $d$ where the total weight placed on prices above $x$ (where $x \ge \frac{3}{k}$) is equal to $b$, for $b \leq \frac{1}{24 k }$. Then, for the other seller (seller 1), there is some price $x' < x$ such that $u_1(x',d) \ge u_1(x,d) + \frac{1}{24k^2}.$
\end{lemma}


The above lemma has a simple corollary which will also be useful:

\begin{corollary} In a Bertrand stage game in which one seller's distribution is $d$, pricing at $p \geq max(supp(d))$ is a strictly dominated strategy when $\max(supp(d)) \geq \frac{3}{k}$. \label{corr:max_dominated}
\end{corollary}

We are now ready to show what constitutes competitive prices in the Bertrand game.

\begin{lemma}
\label{lemma:unique_ne}
    All Nash Equilibria of the Bertrand stage game are supported on prices that are at most $\frac{2}{k}$. 
\end{lemma}

\begin{proof}
We will show this by induction:

\noindent\emph{Base Case:} No Nash Equilibrium (NE) of this game have support on price $1$. To see this, note that by Corollary~\ref{corr:max_dominated}, the price $1$ is dominated for both sellers, and thus cannot be in the support of any NE. 

\noindent\emph{Induction Hypothesis:}  If no NE has support on price $\frac{m}{k}$ or higher and $m$ is an integer s.t. $m \geq 4$, then no NE can have support on price $\frac{m-1}{k}$. To see this, consider any NE, $(d_{1},d_{2})$ such that neither distribution is supported on prices $\frac{m}{k}$ or higher. Since $max(supp(d_1)) \leq \frac{m-1}{k}$, where $\frac{m-1}{k} \geq \frac{3}{k}$, the best-response to it cannot be $\frac{m-1}{k}$ (owing to Lemma~\ref{app:marginal_domination}). 
Therefore $d_{2}$ cannot have support on price $\frac{m-1}{k}$. A similar argument holds for $d_1$, extending the induction hypothesis.

Therefore, no Nash Equilibria of the game have support on any prices above $\frac{2}{k}$. 
\end{proof}

We observe a gap between the equilibrium prices of the stage game between the Nash and the Stackelberg equilibrium. The following lemma relies upon results we will show in Section~\ref{sec:supra}, but we state it here for ease of organization.

\begin{lemma}
    The Stackelberg Equilibria of the stage game leads to supra-competitive prices
\end{lemma}
\begin{proof}
    By Lemma~\ref{lemma:uniform_is_good_enough}, the optimizer can achieve utility $\Omega_{k}(T)$ against a no-swap regret learner (by playing a static strategy).~\cite{deng2019strategizing} establish that the maximum possible average payoff achievable against a no-swap regret learner is upper bounded by the stage game Stackelberg leader value. Therefore, the stage game Stackelberg leader value is $\geq \Omega_{k}(1)$, and thus the resulting price is $\geq \Omega_{k}(1)$ (see Observation~\ref{obs:value_transfer} which lower bounds the average price by the average payoff of either seller). 
\end{proof}

\subsection{Results for the Repeated Bertrand Game}

We begin by considering the Stackelberg equilibria of the repeated Bertrand game. As we observed earlier, allowing unrestricted commitment to any possible algorithm opens the door for explicit usage of threats to maintain collusive prices. The following lemma shows that supra-competitive prices emerge out of such commitments, and is unsurprising.

\begin{lemma}
    \label{se_not_the_same_as_ne}
    The Stackelberg Equilibrium of the repeated Bertrand game induces supra-competitive prices.
\end{lemma}

The proof, which is based on the algorithm of~\cite{collina2023efficient} to find optimal algorithmic Stackelberg strategies in repeated games, is in Appendix~\ref{app:se_proof}. It is known that the resulting leader algorithm is ``obviously" anti-competitive, in that it encodes an explicit threat for the follower. However, we show that supra-competitive prices arise from a vast array of algorithms which are not facially anti-competitive. In particular, we consider the class of mean-based no-regret algorithms and the class of no-swap regret algorithms, which are well-behaved in the following sense: for both classes of algorithms, when both sellers use any algorithm in the class, the prices converge to competitive prices.

\subsubsection{Committing to No-Regret Algorithms Induces Supra-Competitive Prices} 

\label{sec:supra}
In this section, we will prove two key results pertaining to the repeated version of the game:

\begin{enumerate}
    \item While playing against any no-regret learner, the optimizer can always guarantee themselves at least $\Omega_{k}(T)$ net payoff, even when restricted to static non-responsive strategies (recall: strategies which are oblivious to their opponent's behavior and simply play the same pricing distribution every day). In fact, they attain $\Omega_{k}(T)$ net payoff even by playing a uniform distribution over prices each day.
    \item If an optimizer attains $\Omega_{k}(T)$ net payoff against a no-regret Learner, that no-regret Learner also obtains $\Omega_{k}(T)$ net payoff.
\end{enumerate}

Taken together, these two results imply that, if the optimizer, optimizing in the environment defined by the learner's no-regret algorithm, does at least as they could do using a simple static and non-responsive strategy, the induced prices are supra-competitive. Furthermore, \emph{both} the learner and the optimizer benefit from these high prices. Finally, we will show that, when the learner is playing a no-swap regret algorithm and the optimizer plays a static non-responsive best response, this is in fact an approximate \emph{Nash} Equilibrium in algorithm space (where the approximation is a sublinear additive factor). We emphasize that despite the fact that the optimizer happens to be playing a static, non-responsive algorithm, this is a best response for them in all of (unrestricted) algorithm space. This result in fact holds true in all games not just pricing games---but in the context of pricing games, it gives an example of a Nash equilibrium with supra-competitive prices between two algorithms which do not explicitly encode threats.

To achieve these results, we must reason about an optimizer's payoff when playing against a no-regret learner. Consider a no-regret algorithm playing against a seller who plays the same distribution in each round. This fixes the expected cumulative payoff of each fixed action for the no-regret learner; and thus, by the no-regret guarantee, the learner must play a best-response for all but a diminishing fraction of rounds. This, in turn, has implications for the follower's aggregate payoff. We capture this relationship in the following lemma. Note that this lemma applies to all bimatrix repeated games, not just pricing games.
\begin{lemma} 
\label{lemma:fixed_strategy_nr}
Consider any $T$-round repeated game defined by a stage game $G$ with non-negative payoffs. Against any learning algorithm with the no regret property, if the optimizer plays a static strategy $s$ each round, the optimizer's expected payoff is at least $u_{o}(BR^{*}_{\ell}(s),s) \cdot T - o(T)$, where $u_{o}$ is the optimizer's value in $G$ and $BR^{*}_{\ell}(s)$ is the element of the learner's best response to $s$ in $G$ that minimizes $u_{o}(BR^{*}_{\ell}(s),s)$. \label{lem:static}
\end{lemma}

\begin{proof}

Let the total regret of the learning algorithm be $r(T)$, and let $BR_{\ell}(s)$ be the set of all best responses to $s$ from the learner. Furthermore, let $g = u_{\ell}(BR^{*}_{\ell(s)},s) - \max_{y \in [m] \backslash BR_{\ell(s)}}u_{\ell}(y,s)$. 

Let $p^{*}$ be the average probability distribution played by the learning algorithm against $s$. Since there is no correlation between the two seller's randomness, the learner's average payoff is

\begin{align*}
 u_{\ell}(p^{*},s) & \le \Pr[p^{*} = x \in BR_{\ell}(s)] \cdot u_{\ell}(BR_{\ell}(s),s) + \Pr[p^{*} = x \notin BR_{\ell}(s)] \cdot (u_{\ell}(BR_{\ell}(s),s) - g)
\end{align*}

As the learner's algorithm has average regret $\frac{r(T)}{T}$, we have that 
\[
\Pr[p^{*} = x \in BR_l(s)] \cdot u_{\ell}(BR_{\ell}(s),s) + \Pr[p^{*} = x \notin BR_{\ell}(s)] \cdot [u_{\ell}(BR_{\ell}(s),s) - g] \geq u_{\ell}(BR_{\ell}(s),s) - \frac{r(T)}{T}.
\]
Therefore, 
\[
\Pr[p^{*} = x \notin BR_{\ell}(s)] \cdot  g \leq  \frac{r(T)}{T},
\]
which implies
\[
\Pr[p^{*} = x \notin BR_{\ell}(s)]  \leq  \frac{r(T)}{Tg},
\]
and finally 
\[
\Pr[p^{*} = x \in BR_{\ell}(s)]   \geq 1 -  \frac{r(T)}{Tg}.
\]

This allows us to lower bound the payoff of the optimizer  (where the first inequality implicitly uses the fact that the both sellers always get non-negative payoff) --

\begin{align*}
  u_{o}(p^{*},s) & \geq \Pr[p^{*} = x \in BR_{\ell}(s)] \cdot u_{o}(BR_{\ell}^{*}(s),s) \\
  & \geq \left(1 - \frac{r(T)}{Tc}\right)u_{o}(BR_{\ell}^{*}(s),s) \\
& = u_{o}(BR_{\ell}^{*}(s),s) - \frac{r(T)}{T} \cdot \frac{u_{o}(BR_{\ell}^{*}(s),s)}{g} \\
& = u_{o}(BR_{\ell}^{*}(s),s) - \frac{o(T)}{T} \cdot \frac{u_{o}(BR_{\ell}^{*}(s),s)}{g} \\
& = u_{o}(BR_{\ell}^{*}(s),s) - \frac{o(T)}{T}.
\end{align*}
\end{proof}


We are now ready to show that if an optimizer's only goal is maximizing their own payoff, and they do even a passable job, they will get high payoff against a no-regret learner. In fact, they can do this even by performing the extremely simple static strategy of pricing uniformly randomly every day. Recall that by Observation~\ref{obs:value_transfer}, the average payoff of the optimizer lower bounds the average price. Therefore, as long as the optimizer is optimizing in the environment defined by the learner's no regret algorithm, there is a wide range of optimizer behaviors which induce supra-competitive prices. This includes all optimizers that perform better than static random pricing, ranging from approximate optimization over static responses to exact optimal dynamic responses. 

\begin{lemma} 
\label{lemma:uniform_is_good_enough}
The optimizer can get payoff at least $\Omega_{k}(T)$ against any no-regret learner in the Bertrand model by playing a static strategy of uniformly random pricing. Formally, for any no-regret algorithm $\cA^{noreg}$ and the static uniformly random algorithm $\cA^{r}$,  \begin{align}
    U_{o}(\cA_{l}^{noreg},\cA_{o}^{r}) \geq \Omega_{k}(T).
\end{align}\label{lem:random} 
\end{lemma}

\begin{proof}
Let $r$ be the uniformly random distribution over prices.
By Lemma~\ref{lem:static}, the optimizer's expected payoff is at least $u_{2}(BR^*_{\ell}(r),r) \cdot T - o(T)$. 

First, we will prove that any best-response of the learner $BR_{\ell}(r) \geq \frac{1}{5}$. To see this, note that by pricing at $\frac{1}{2}$ against $r$, the learner's payoff would be at least

\begin{align*}
 \Pr \left [r > \frac{1}{2} \right ] \cdot \frac{1}{2} 
  & = \Pr \left [r \geq \floor{\frac{1}{2} + \frac{1}{k}} \right ] \cdot \frac{1}{2}  
  \\ & =  \frac{1 - \floor{\frac{1}{2} + \frac{1}{k}}}{1 - \frac{1}{k}} \cdot \frac{1}{2} 
\\ & \geq  \left(1 - (\frac{1}{2} + \frac{1}{k})\right)\cdot \frac{1}{2} 
\\ & \geq  \frac{\frac{1}{2} - \frac{1}{k}}{2} 
\\ & \geq \frac{1}{4} - \frac{1}{2k}
\\ & \geq \frac{1}{4} - \frac{1}{20} = \frac{1}{5} 
\end{align*} 

As the learner can get payoff at least $\frac{1}{5}$, their best response price will always be at least $\frac{1}{5}$, as otherwise they would always get payoff strictly less than $\frac{1}{5}$.

We now invoke Lemma~\ref{lemma:fixed_strategy_nr} to lower bound the payoff of the optimizer. The worst learner best-response to the uniform distribution, from the perspective of the optimizer, must be bounded below by $\frac{1}{5}$ i.e. the optimizer must select a price that is at least $p^* := \frac{\floor{\frac{k}{5}}}{k}$~\footnote{By construction, $p^*$ is in the set of prices available to the seller}. Additionally, the optimizer's payoff is  monotone increasing in the learner's price. Thus the optimizer gets total payoff at least $T \cdot u_o\left(r,p^*\right) - o(T)$. Analyzing this gives us the desired result --

\begin{align*}
u_{o}(r, p^*) \cdot T - o(T) 
& \geq  \mathbb{E}[r | r < p^*] \cdot \Pr[r < p^*] \cdot T - o(T) 
\\ & \geq  \mathbb{E}\left[r | r \leq p^*-\frac{1}{k}\right] \cdot \Pr\left[r \leq p^*-\frac{1}{k}\right]\cdot T - o(T) 
\end{align*}

$\mathbb{E}[r | r <= p^*-\frac{1}{k}]$ is exactly $\frac{p^*}{2}$ (the average of $\frac{1}{k}$ and $p^* - \frac{1}{k}$). Substituting $p^*$ results in $\frac{p^*}{2} = \frac{\floor{\frac{k}{5}}}{10}$, this is lower bounded by $\frac{2}{25}$ (assuming $k >20$).  On the other hand, we can rewrite $\Pr\left[r \leq p^*-\frac{1}{k}\right]$ (by substituting $p^*$) as $\frac{\floor{\frac{k}{5}}-1}{k}$. This term is monotone increasing in $k$ and lower bounded by $\frac{3}{25}$ (assuming $k > 20$). Thus, the optimizer payoff is lower bounded by $\frac{6T}{625}  - o(T)$, completing the proof.

\end{proof}

Since by Observation~\ref{obs:value_transfer} the average payoff of the optimizer lower bounds the average price, we have already shown that any no-regret algorithm induces supra-competitive prices against any optimizer who performs better than random pricing. We will now go on to show that the learner deploying the no-regret algorithm benefits from these supra-competitive prices as well.  

\begin{theorem}
\label{thm:equal_looting}
   For any no-regret learner algorithm $\cA^{noreg}_{l}$, and for any optimizer algorithm $\cA_{o}$, if $$U_{o}(\cA^{noreg}_{l}, \cA_{o}) \geq \Omega_{k}(T)$$ then $$U_{l}(\cA^{noreg}_{l},\cA_{o}) \geq \Omega_{k}(T)$$
\end{theorem}
\begin{proof}
Let $U_{o}(\cA^{noreg}_{l}, \cA_{o}) \geq c \cdot T$ for some constant $c$ between $0$ and $1$, independent of $k$. 

Note that, in order to achieve $c$ payoff on average, the optimizer must have priced at or above $\frac{c}{2}$ with average probability (over the rounds) at least $\frac{c}{2}$. Assume for contradiction that the optimizer does not do this. Then, we can upper bound their utility by assuming they capture demand in all rounds:  
\begin{align*}
U_{o}(\cA^{noreg}_{l}, \cA_{o}) & \leq 
\Pr\left[p_{o} \geq \frac{c}{2}\right] \cdot 1 + \left(1 - \Pr\left [p_{o} \geq \frac{c}{2}\right] \right) \cdot \left (\frac{c}{2} \right) \\ 
& < \frac{c}{2}  + \left (1-\frac{c}{2} \right) \cdot \frac{c}{2} \\
& = c - \frac{c^{2}}{4}
\end{align*}
This is strictly less than $c$, and thus by contradiction, the optimizer must have priced at or above $\frac{c}{2}$ with frequency at least $\frac{c}{2}$. 

Given this, one fixed action that the learner could have taken is to price at $\frac{c}{2}$. Then,

\begin{align*}
U_{l}(\cA^{noreg}_{l}, \cA_{o})
& \geq \frac{1}{2}\Pr\left[p_{o} \geq \frac{c}{2} \right] \cdot \left (\frac{c}{2} \right ) \\ &
 \geq \frac{1}{2} \cdot \frac{c}{2} \cdot \left (\frac{c}{2} \right) 
 = \frac{c^{2}}{8} 
\end{align*}

By the no-regret guarantee of the learner, therefore, we have that  $U_{l}(\cA^{noreg}_{l}, \cA_{o}) \geq \frac{c^{2}T}{8} - o(T) = \Omega_{k}(T)$.

\end{proof}

We can now combine these results, to show that, if the learner deploys any no-regret algorithm in the Bertrand model and the optimizer responds via any strategy which gets them payoff at least that of static random pricing, then prices are supra-competitive and both the learner and the optimizer get a constant fraction of the profits.

\begin{theorem} In a Bertrand repeated game, for any no-regret learner algorithm $\cA_{l}^{noreg}$ and any optimizer algorithm $\cA_{o}$ such that 
\begin{align}
    U_{o}(\cA_{l}^{noreg}, \cA_{o}) \geq U_{o}(\cA_{l}^{noreg}, \cA_{o}^{r}),
\end{align}
where $\cA_{o}^{r}$ is the static uniformly random algorithm, we have
\begin{align}
    U_{o}(\cA_{l}^{noreg}, \cA_{o}) = \Omega_{k}(T)
\end{align}
 
\begin{align}
    U_{l}(\cA_{l}^{noreg}, \cA_{o}) = \Omega_{k}(T)
\end{align}
Further, the average price is $\Omega_{k}(1)$.
\end{theorem}
\begin{proof}
    By Lemma~\ref{lem:random}, the optimizer gets payoff at least $\Omega_{k}(1)\cdot T$ against any no-regret algorithm by playing static random pricing. Therefore, the optimizer's payoff utilizing any strategy which is better than static random pricing also gives them a payoff of at least $\Omega_{k}(1)\cdot T$.
    Furthermore, by Theorem~\ref{thm:equal_looting}, this implies that the learner also gets a payoff of at least $\Omega_{k}(1)\cdot T$. Finally, by Observation~\ref{obs:value_transfer}, this implies that the average price is $\Omega_{k}(1)$.
\end{proof}

We show an equivalent result for the more general Multinomial-logit model in Theorem~\ref{thm:logit_generalized}.


Note that despite our treatment of the problem in a sequential play setting, neither sequential play nor commitment power is necessary for our results. In fact, in the game where leader and optimizer pick their algorithm simultaneously, if one player plays any no-swap regret algorithm and their opponent plays a static strategy corresponding to the Stackelberg leader distribution of the stage game, then this forms a Nash equilibrium in algorithm space:

\begin{theorem} There exists an $o(T)$-approximate Nash Equilibrium in algorithm space in any repeated game consisting of a no-swap regret algorithm for the leader and a static, non-responsive algorithm for the follower.
\end{theorem}

\begin{proof}
In particular, such an equilibrium exists between any algorithm with total swap regret $r(T) = o(T)$ and the static algorithm which plays (something very close to) the Stackelberg leader strategy of the stage game each round. To show that this is true, we will show that each is an $\frac{o(T)}{T}$-approximate best response to the other. 

First, we will show that a particular distribution (which is $\varepsilon$-close to the Stackelberg leader strategy of the stage game) is near-optimal against any No-Swap Regret algorithm (NSR) when played on every round. By Lemma~\ref{lemma:nsr_bestresponse}, there exists a distribution $D'$ such that $||D-D'||_\infty \le \varepsilon$, where $D$ is the static Stackelberg leader strategy and $D'$ is a $o(T)$-best response to any algorithm with sublinear swap regret $r(T)$. 
     
It remains to show that NSR is optimal against this distribution. However, NSR is optimal against any static distribution played on every round. The best response against any algorithm playing a static distribution is to play (one of) the best-response action(s) each day. Therefore, the gap between the average payoff of the optimal response and the average payoff attained by the NSR algorithm is bounded by the average external regret of the NSR algorithm, which is at most $r(T) = o(T)$. 
\end{proof}

Here, notably, both players are ``succeeding in optimizing'' since they are playing best responses in algorithm space against one another, neither is using an algorithm which encodes threats, and yet the outcome is near monopoly prices.

\subsection{(Many) No-Regret Dynamics Converge to Competitive Prices When Used by Both Sellers}


The main result of this subsection is showing that both sellers using mean-based no-regret algorithms (which include all algorithms in the follow-the-regularized-leader family such as hedge, FTPL) results in convergence to competitive prices. This complements \cite{hartline2024regulation} who show a similar result for the class of no swap regret algorithms. 

\begin{theorem}
\label{theorem:mbnr_convergence}
If both sellers set prices using a $\frac{1}{\sqrt{T}}$-mean-based no-regret algorithm~\footnote{The $\frac{1}{\sqrt{T}}$ mean-basedness is chosen because fixing the optimal static learning rate for algorithms such as multiplicative weights and follow-the-perturbed leader results in this property. 
}, after $O(k^{3k} \sqrt{T \log k})$ rounds the sellers will pick competitive prices with probability $(1-o_T(1))$ in each subsequent round of play 
regardless of the initial distribution played by either algorithm. 
\end{theorem}

The proof of this theorem directly follows from the application of the lemma below.

\begin{lemma}
\label{lemma:bertrand_mbnr_convergence}
    After $t \ge 24 k^{3(k-i+1)} \sqrt{T \log k}$ rounds, both $\frac{1}{\sqrt{T}}$-mean-based no-regret algorithms pick prices $\frac{i}{k}$ or higher with an aggregate probability of at most $\frac{(k-i)}{k^{2.5}}$ for all $i > 3$.
\end{lemma}

\begin{proof}

    [Base Case : $i=k$] We begin by showing that after $24 k^3\sqrt{T}$ rounds, the MBNR algorithms learn not to play action $1$ with probability greater than $\frac{\log k}{\sqrt{T}}$. A direct consequence of Lemma~\ref{app:marginal_domination} is that action $1$ lags the leader action in terms of maximum payoff (in hindsight) by $\sqrt{T}$ after these rounds, since action $1$ is worse by $\frac{1}{24k^2}$ in payoff against the time average distribution of the other player at the end of these rounds as compared to some action $x' \in \cP$ with $x' < 1$.
    It suffices to extend the induction hypothesis from the perspective of one of the MBNR players, WLOG Player 1.

    Assuming the proof holds up to a given $i$, we extend it for $i-1$ using the following idea. The outline is that we start at time $t'$, beyond which prices $\frac{i}{k}$ and above are guaranteed to be picked with low probability on a round-to-round basis. Then, we magnify $t'$ by a factor of $k^3$ to get $t''$, so as to argue that the average distributions of either player beyond this point has negligible probability on prices $\frac{i}{k}$ and above. This allows us to use Lemma~\ref{app:marginal_domination} to argue that at $t \ge t''$, price $\frac{i-1}{k}$ is much worse in hindsight than some other price, implying that $\frac{i-1}{k}$ is also henceforth picked with negligible probability on a round-to-round basis, successfully extending the induction hypothesis. The exact argument is given below.


    On any round $t$ with  $t \ge t' = 24k^{3(k-i+1)} \sqrt{T \log k}$; we know that the MBNR algorithm of Player 2 picks prices above $\frac{i}{k}$ with probability at most  $\frac{(k-i)}{k^{2.5}}$. Now pick $t'' = k^3 t'$ and consider the marginal distribution of player 2 at this point. This distribution, based on our induction hypothesis puts probability mass at most $\frac{t' + \frac{(k-i)}{k^{2.5}}k^3 t'}{k^3t'} \le \frac{1}{24k}$ on or above price $\frac{i}{k}$. The same holds true for time steps $t \ge t''$.

    For large enough $k$ and $i > 3$, we can now invoke Lemma~\ref{lemma:marginal_domination} with $b = \frac{1}{24k}$ and $x = \frac{i-1}{l}$ to argue that for Player 1, the action $\frac{i-1}{k}$ obtains at least $\frac{1}{24k^2}$ less in payoff as compared to some price $x' < \frac{i-1}{k}$
    (for large $k$) against Player 2's time average distribution after $t\ge t''$ rounds. This implies that the aggregate lag of the payoff of action $\frac{i-1}{k}$ against the leader action in hindsight after $t \ge t''$ rounds is at least $\frac{t''}{24k^2} 
    \geq \sqrt{T}\log(k)$. Thus we can invoke the mean based property implying action $\frac{i-1}{k}$ is henceforth (for $t \ge t''$) picked by Player 1 with probability at most $\sqrt{\frac{\log k}{T}}$. Thus, the total probability picking price $\frac{i-1}{k}$ or higher is upper bounded by $\frac{(k-i)}{k^{2.5}} + \sqrt{\frac{\log k}{T}}$, which is upper bounded by $\frac{(k-i +1)}{k^{2.5}}$ for sufficiently large $T$, completing the extension of the induction hypothesis.
\end{proof}

We also provide an elementary proof of a similar result about no-swap-regret algorithms. However, these are not new results, having been established already for a more general model by~\citet{hartline2024regulation} and~\citet{chassang23}. We include these proofs in Appendix~\ref{app:nsr_convergence}.

\section{Numerical Investigation of Constants}
\label{sec:experiments}


We showed in Lemma~\ref{lemma:uniform_is_good_enough} that the optimizer can get a payoff of $\Omega_k(T)$ against any no-regret learner, just by playing the uniform distribution over their prices. An implication is that they would do at least as well with their optimal static strategy, which was enough for our results which are stated in asymptotic notation. While we did not get attractive constants via this simple approach, we numerically verified that the optimal static optimizer distribution against a no-regret algorithm, which is the Stackelberg leader strategy for a number of values of $k$, results in average prices of more than $\frac{2}{3}$ (in fact approximately $\frac{2(k-1)}{e k}$). ~\cite{conitzer2006computing} give an efficient algorithm for computing the Stackelberg equilibrium of a two player game, which we use to numerically compute the Stackelberg leader strategy as well as the payoff of the Stackelberg leader, which will approximately (up to subconstant in $T$ additive error) be the average payoff of the optimizer repeatedly playing the Stackelberg leader strategy (using Lemma~\ref{lemma:fixed_strategy_nr}). Our numerical results show that the Stackelberg value of the leader (as well as the follower) is approximately $\frac{k-1}{ek}$ for a range of values of $k$ in Figure~\ref{fig:payoff_growth}. By Observation~\ref{obs:value_transfer}, if the follower optimizes over the space of static strategies, the average price for the buyer will be $\frac{2(k-1)}{ek} \geq \frac{2}{3}$. We note that while our experiments show these prices are achievable by non-responsive optimizers for all values of $k$ from $1$ to $200$, it remains interesting future work to prove that this holds for all $k$. Finding a closed form analytical expression for the Stackelberg leader strategy of the stage game is an open problem and would likely help in proving the conjecture for all  $k$.

\begin{figure}[h]
    \centering
    \begin{subfigure}[b]{0.48\textwidth}
        \centering
        \includegraphics[width=\textwidth]{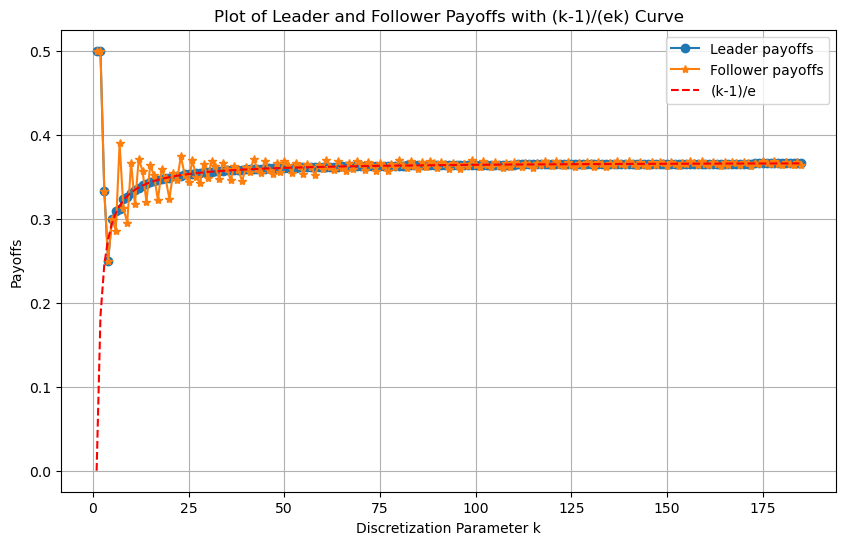}
        \caption{Stackelberg Leader/Follower Payoffs for varying $k$}
        \label{fig:payoff_growth}
    \end{subfigure}
    \hfill
    \begin{subfigure}[b]{0.48\textwidth}
        \centering
        \includegraphics[width=\textwidth]{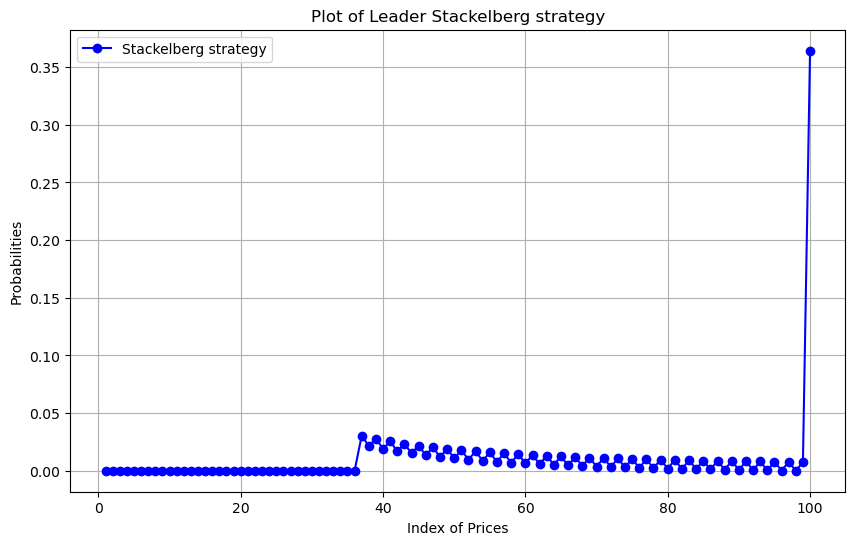}
        \caption{Stackelberg Leader Strategy for $k=100$.}
        \label{fig:stack_100}
    \end{subfigure}
    \label{fig:combined_stackelberg}
\end{figure}

We also show the Stackelberg strategy for $k=100$ in Figure~\ref{fig:stack_100}. Our numeric results also indicate that every price between $36/100$ to $98/100$ is (tied for being) a best response for the follower given the optimal leader commitment\footnote{In the Stackelberg Equilibrium problem, the follower typically tie-breaks in favor of the leader, in our application the optimizer adds some infinitesimal extra probability mass on price $ 99/100$ to ensure that no-regret learner learns a unique best-response of $98/100$ (which is the best outcome for the optimizer).}.

The code for these experiments can be found at \href{https://github.com/eshwarram/Non_Myopic_Pricing.git}{Code Repository}.

\section{Discussion and Conclusion}
Defining anti-competitive behavior is delicate. Monopoly like prices can arise in a number of scenarios: a \emph{failure to optimize} (if e.g. players are simply not best responding to one another) as well as through \emph{collusion} (which has been interpreted as requiring explicit \emph{threats} when speaking of algorithmic collusion \cite{calvano20,harrington}). Past work \cite{hartline2024regulation,chassang23} have proposed no-swap-regret algorithms as reasonable competitive algorithms in pricing scenarios, on the basis that they converge to competitive prices when both parties use them, and they seem both to successfully optimize and not to encode threats. 

Our results complicate this picture. We show that if the first entrant into a market deploys a pricing algorithm with the no-swap-regret guarantee, then this very strongly incentivizes the next entrant to deploy an algorithm that will lead to supra-competitive prices. In fact, \emph{anything the second entrant does} that obtains them profit at least that of a random pricing strategy will inevitably lead to supra-competitive prices. Moreover, this will not be at the expense of the first entrant---the no(-swap) regret learner will \emph{also} enjoy supra-competitive revenue. And the algorithm of the second entrant can be entirely static and non-responsive, and therefore unable to encode threats. In fact, this phenomenon does not hinge on the sequential nature of play that we focus on and does not hinge on either player having commitment power. As we show, there is a Nash equilibrium of the game---in algorithm space---maintaining supra-competitive prices, which involves one player playing a no swap regret algorithm and the other playing a static pricing distribution. Both players are best responding to one another in the space of all pricing algorithms (without any restriction) and so neither player is failing to optimize---but neither are either of the players deploying threats. We suggest that this might require a reconsideration of what constitutes algorithmic collusion.

\subsection*{Acknowledgments}

We thank Rakesh Vohra and Deke Hill for valuable discussions on the subject of algorithmic collusion.

SK was on leave from the University of Pennsylvania and serving as the Associate Director of the Simons Institute for the Theory of Computing at the time of writing of this paper.

ERA was  supported by NSF grants CCF 1910534 and 2045128. NC and AR were partially supported by NSF grants FAI-2147212 and CCF-2217058, the Hans Sigrist Prize, and the Simons Collaboration on Algorithmic Fairness.

\bibliography{pricing_refs}
\bibliographystyle{plainnat}

\appendix
\section{Assorted Results and Proofs - Bertrand Duopoly Model}

\begin{lemma}
An optimizer playing against a no regret algorithm cannot get utility greater than $\Omega(4(\sqrt{\frac{3}{2}} - 1)kT < \Omega(\frac{9}{10}kT)$.
\end{lemma}

\begin{proof}
By Theorem~\ref{thm:equal_looting}, if the total utility of the optimizer is $ckT$, then the total utility of the no regret learner is at least $\frac{c^{2}kT}{8}$. Given that the maximum possible welfare each round is $k$, the maximum possible welfare over the entire game is $kT$, and therefore
\begin{align*}
& ckT +  \frac{c^{2}kT}{8} \leq kT \\
& \implies \frac{1}{8}c^{2} + c - 1 \leq 0\\
& \implies c \leq 4(\sqrt{\frac{3}{2}} - 1) < \frac{9}{10} \\
\end{align*}
\end{proof}

\subsection{Proof of Lemma~\ref{lemma:marginal_domination}}
\label{app:marginal_domination}

Let the distribution of the second seller be $d$. 
There are two cases:

\begin{itemize}

\item Case 1:  $\Pr[d = x] \geq \frac{1}{6k}$. 
Since $x \geq \frac{3}{k}$, it is possible for the first seller to price at $x - \frac{1}{k}$. The price $x-\frac{1}{k}$ gives the first seller utility at least $\Pr[d = x]\cdot (x-\frac{1}{k}) + \Pr[d > x]\cdot (x-\frac{1}{k})$, while action $x$ gives utility $\frac{\Pr[d = x]\cdot x}{2} + \Pr[d > x]\cdot x$. The gap between these utilities is
\begin{align*}
& \Pr[d = x]\cdot \left(x-\frac{1}{k}\right) + \Pr[d > x]\cdot \left(x-\frac{1}{k}\right) - \frac{\Pr[d =x] \cdot  x}{2} - \Pr[d > x]\cdot x\\
& =  \Pr[d = x]\cdot \left(\frac{x}{2}-\frac{1}{k}\right) - \frac{\Pr[d > x]}{k}\\
& \geq \frac{1}{6k} \cdot \left(\frac{x}{2}-\frac{1}{k}\right) -  \frac{\Pr[d > x]}{k}\\
& = \frac{1}{6k} \cdot\left(\frac{x}{2}-\frac{1}{k}\right) - \frac{b}{k} \\
&\geq \frac{1}{12k^2} - \frac{b}{k} \\
&\geq \frac{1}{12k^2} - \frac{1}{24k^2}
=\frac{1}{24k^2} \\
\end{align*}

\item Case 2: $\Pr[d = x] < \frac{1}{6k}$. Then, the utility for the second seller of playing $x$ is at most $\frac{1}{6k} \cdot \frac{x}{2} + \Pr[d > x]\cdot x \leq x \cdot \left (\frac{1}{12k} + b \right) \leq 
\frac{x}{k}\cdot \left(\frac{1}{12} + \frac{1}{24} \right) \leq \frac{1}{6k}$.  
Then, note that the payoff for the second seller of pricing at $\frac{1}{k}$ every round is at least $\frac{1}{2k}$. Thus, the gap in payoffs between action $x$ and action $\frac{1}{k}$ is at least $\frac{1}{2k} - \frac{1}{6k} = \frac{1}{3k}$. 
\end{itemize}

\subsection{Proof of Lemma~\ref{se_not_the_same_as_ne}}
\label{app:se_proof}
\begin{proof}
We utilize the framework introduced in~\cite{collina2023efficient} to find the Stackelberg equilibrium of this repeated game. Consider the following leader algorithm for some interaction of length $T$: from rounds $1$ to $\bar{t} = T - \frac{T}{2(\frac{k-1}{k})}- 1$, the leader asks the follower to price at $1$, while the leader prices at $1 - \frac{1}{k}$. If the follower does so, the leader will price at $1$ for the remaining rounds. However, if at any point the follower deviates to a price that is not $1 - \frac{1}{k}$ in these initial $\bar{t}$ rounds, the leader will deviate to pricing at $\frac{1}{k}$ for the remainder of the game.  

Let time $t^{*}$ be the first round in which the follower prices at some price which is not $1$, and let that price be $p_{t^{*}}$. If $t^{*} \leq \bar{t}$, then the utility of the follower is at most
\begin{align*}
    & (\sum_{t=1}^{t^{*}-1}0) + p_{t} + (\sum_{t=t^{*}+1}^{T}\frac{\frac{1}{k}}{2})
    \\ & \leq (1 - \frac{2}{k}) + \sum_{t=t^{*}+1}^{T}\frac{1}{2k}
    \\ & \leq 1 - \frac{2}{k} + \frac{T}{2k}
\end{align*}

If $t^{*} = \bar{t} + 1$, the follower can price at $1 - \frac{1}{k}$ for the remaining rounds. Then their utility is 

\begin{align*}
    & (\sum_{t=1}^{\bar{t}}0) +(\sum_{t=\bar{t}+1}^{T} (1 - \frac{1}{k}))
    \\ & =(T - \bar{t})(1 - \frac{1}{k})
    \\ & =(\frac{T}{2((1 - \frac{1}{k}))}+1)(1 - \frac{1}{k})
    \\ & =\frac{T}{2}+ 1 - \frac{1}{k}
\end{align*}

This is strictly larger than the best utility that the follower can get if $t^{*} \leq \bar{t}$. Therefore, the best response by the follower will involve them pricing at $1$ for the first $\bar{t}$ rounds. Thus the utility of the leader is at least

\begin{align*}
    (1 - \frac{1}{k}) \cdot \bar{t} 
     & = (T - \frac{T}{2(1 - \frac{1}{k})}- 1) (1 - \frac{1}{k})
     \\ & = (T-1)(1 - \frac{1}{k}) - \frac{T}{2}
    \\ & \geq (T-1)(\frac{19}{20}) - \frac{T}{2} \tag{As $k \geq 20$}
    \\ & \geq \frac{3T}{4}(\frac{19}{20}) - \frac{T}{2} \tag{For $T \geq 4$}
    \\ & = \frac{17}{80}T = \Omega_{k}(T)
\end{align*}

This particular commitment algorithm gets the leader average utility of $\Omega_{k}(1)$. Therefore, the optimal Stackelberg commitment algorithm in the Bertrand pricing game gets the leader utility of $\Omega{k}(1)$. Finally, by Observation~\ref{obs:value_transfer}, this lower bounds the average price, and therefore the average price is also $\Omega{k}(1)$. 
\end{proof}

\subsection{NSR Convergence}

\label{app:nsr_convergence}

\begin{theorem}
\label{thm:nsr_vs_nsr}
    If both players use no swap-regret algorithms, then the prices converge to competitive prices.
\end{theorem}

The proof of this result follows from the lemma below.
\begin{lemma}
\label{lemma:ce_competitive}
    All Correlated Equilibria are supported on prices that are at most $\frac{2}{k}$.
\end{lemma}

\begin{proof}
We will show this by induction. 

Base Case: No Correlated Equilibria of this game have support on price $p_k$. To see this, note that by Corollary~\ref{corr:max_dominated}, price $1$ is dominated for both players in the CE. Therefore, no CE can have support on $1$.   

Inductive Hypothesis: If no CE has support on price $\frac{m}{k}$ for some integer $m$ and $m \geq 4$, then no CE can have support on price $\frac{m-1}{k}$. Assume for contradiction that this is not the case, and consider some CE in which one of the players is recommended price $\frac{m-1}{k}$. Consider the distribution $d$ of the other player conditional on this player being recommended $\frac{m-1}{k}$. As no CE has support on any at or above $\frac{m}{k}$, $max(supp(d)) \leq \frac{m-1}{k}$, where $m-1 \geq 3$. By the definition of being a CE, $\frac{m-1}{k}$ must be a best response to $d$. But by Corollary~\ref{corr:max_dominated}, $\frac{m-1}{k}$ is strictly dominated against $d$. This is a contradiction, and therefore no CE can have support on price $\frac{m-1}{k}$.

Therefore, no correlated equilibria of the game have support on any prices above $\frac{2}{k}$.
\end{proof}

\section{Generalization of Results to Logit Model}

 \label{app:logit}

In this section, we show that our results are robust to variations in the pricing game model by extending to logit allocation rules. 

First, we will show that the competitive price is low, at least when $\tau$ is sufficiently large. 

Fix Seller no. 2 to picking price $p$. Then, the payoff of the first seller for a price $x$ is $x \cdot C_{1}^{L,\tau}$. Let $F_p(x) := x \cdot C_{1}^{L,\tau}$ denote this payoff function (with the domain being extended to the unit interval $(0,1]$). We prove some facts about this payoff and where it is maximized. Define $\pmax(\tau) := \max\{ \frac{2}{k(1-e^{-\tau})}, \frac{2}{\tau} \}$

\begin{lemma}
\label{lemma:payoff_properties}
Given that $p > \frac{2}{\tau}$,
$F_p(x)$ is uniquely maximized at some point $x^*_p$ that is strictly less than $p$. Further, for all $x > x^*_p$; $F_p(x) < F_p(x^*_p)$ and $F'_p(x) < 0$. Additionally, if $p > \frac{2}{k(1-e^{-\tau})}$, $F_p\left(p-\frac{1}{k}\right) > F_p(p)$.
\end{lemma}

\begin{proof}
    Writing out the function $F_p(x) = \frac{x}{e^{\tau(x-p)}+1}$. Differentiating, $F'_p(x) = \frac{e^{\tau(x-p)} + 1 - \tau x e^{\tau(x-p)}}{(e^{\tau(x-p}+1)^2}$. Let $N_p(x)$ denote the numerator of this derivative. Since $p > \frac{2}{\tau}$, $N_p(p) < 0$. On the other hand $N_p(0) = 1- e^{-\tau p} > 0$. Taking the derivative of this numerator gives us $N'_p(x) = - \tau^2 x e^{\tau(x-p)} < 0$ for all $x > 0$. Thus, $N_p(x)$ is decreasing for all $x > 0$, implying that there is some $x^* \in (o,p)$ such that $N_p(x^*) = 0$. Thus, the function $F_p(x)$ is maximized at this $x^*$.

    The second result follows directly by substitution. 
\end{proof}

A consequence of the lemma above is that the best response (amongst the discrete set of prices) to a large enough price $p$ is always $p-\frac{1}{k}$ or lower. We strengthen this result to also hold for best-responses to distributions over prices.

\begin{lemma}
\label{lemma:br_below}

If Seller 2 plays a distribution $D$ that is supported on prices in $\{\frac{1}{k},\frac{2}{k} \cdots p^*\}$ with $p^* > \max\{ \frac{2}{k(1-e^{-\tau})}, \frac{2}{\tau} \}$, then the best response of Seller 1 is a price that is strictly smaller than $p^*$.
\end{lemma}

\begin{proof}
    Let $F_D(x)$ denote the payoff obtained by responding to distribution $D$ with a real valued price $x$. Let distribution $D$ put weight $\alpha_i$ on price  $i \in \{\frac{1}{k},\frac{2}{k} \cdots p^*\}$. Thus, we can rewrite $F_D(x)$ as $F_D(x) = \sum_{i=\frac{1}{k}}^{p^*} \alpha_i F_i(x)$. Lemma~\ref{lemma:payoff_properties} tells us that there exists a $x^* < p^*$ such that for all $x > x^*$, $F_i(x) < F_i(x^*)$ and $F'_i(x) < 0$ for all $i \in \{\frac{1}{k},\frac{2}{k} \cdots p^*\}$. Additionally, the same lemma guarantees that $F_i\left(p^*-\frac{1}{k}\right) < F_i(p^*)$ for all such $i$ (from the second part of the lemma for $i=p^*$ and from the first part of the lemma for all other $i$). Thus, $p^*-\frac{1}{k}$ is a better response than any price greater than or equal to $p^*$, which completes the proof of this lemma.
    \end{proof}

\begin{theorem}
\label{thm:ne_logit}
    All Nash equilibria of the stage game induced by the logit model have strategies supported only on prices up to $\pmax(\tau) + \frac{1}{k}$ where $\pmax(\tau)  =  \max\{ \frac{2}{k(1-e^{-\tau})}, \frac{2}{\tau} \}$.
\end{theorem}

\begin{proof}
    This result follows almost directly from Lemma~\ref{lemma:br_below}. Consider any Nash equilibrium of this game $(D_1, D_2)$ where $D_1$ and $D_2$ are distributions on prices in $\cP$. We prove that the max support element in either distribution is at most $\pmax(\tau)+\frac{1}{k}$. 
    For sake of contradiction, assume that the max support element of either distribution, wlog $D_1$, is $p \ge \pmax(\tau)+\frac{2}{k}$. 
    This would imply that seller 2's best response set can only consist of prices strictly below $p$, by invoking Lemma~\ref{lemma:br_below}. However, then invoking the same lemma from seller 1's best response perspective, $D_1$ must be supported on prices below $p-\frac{1}{k}$ since $D_2$'s max support price is $p-\frac{1}{k}$, this contradicts $p$ being in $D_1$'s support, completing the proof  . 
\end{proof}

As long as $\pmax$ is reasonably small, this gives us the desired result- which we capture in the following corollary. Specifically, we need the temperature parameter to be large enough that we are sufficiently sensitive to a change in price of the order of $\frac{1}{k}$.

\begin{corollary}
    For $\tau = \Omega(k)$ and $k \ge 10$, all Nash equilibria strategies of the stage game defined by the logit model are supported on prices that are $O\left( \frac{1}{k}\right)$
\end{corollary}

We will now prove the same two key results that we proved in Section~\ref{sec:bertrand}, but in the more general logit model. Unlike the results above, these results hold regardless of the value of $\tau$.

\begin{enumerate}
    \item While playing against any no-regret learner~\footnote{We follow the convention in prior work in renaming the leader and follower as the learner and optimizer where the leader plays a no-(swap-)regret algorithm. This is helpful in clarifying the roles of the players, as we will use the fact that the optimizer, playing against a no-swap-regret algorithm, can get their \emph{leader} value in the stage game.}, the optimizer can always guarantee themselves at least $\Omega_{k}(T)$ net payoff, even when restricted to simple, non-responsive strategies. In fact, they attain $\Omega_{k}(T)$ net payoff even by playing a uniform distribution over prices each day.
    \item If an optimizer attains $\Omega_{k}(T)$ net payoff agains a no-regret Learner, that no-regret Learner also obtains $\Omega_{k}(T)$ net payoff.
\end{enumerate}

We start by showing that the uniform price distribution also does well against any no-regret algorithm in the Logit model. Note that this result holds for any $\tau$.

\begin{lemma} For any no-regret algorithm $\cA^{noreg}$ and the static uniformly random algorithm $\cA^{r}$, in the multinomial-logit repeated game,  \begin{equation}
    U_{o}(\cA_{l}^{noreg},\cA_{o}^{r}) \geq \Omega_{k}(T)
\end{equation}\label{lem:random} 
 \label{lem:logit_random}
\end{lemma}
\begin{proof}
By Lemma~\ref{lem:static}, the optimizer gets payoff at least $u_{o}(r,BR_{l}(r)) \cdot T - o(T)$ by playing the static uniformly random strategy. Note that by pricing at $\frac{1}{2}$ against $r$, the learner's payoff in the stage game would be at least

\begin{align*}
\frac{1}{k}  \sum_{i=1}^{k} \frac{e^{\tau p_{i}}}{e^{\frac{\tau}{2}}+e^{\tau p_{i}}}
& = \frac{1}{2k}  \sum_{i = 1}^{k} \frac{1}{e^{\tau(\frac{1}{2}-p_{i})}+1}
\\ & > \frac{1}{2k} \sum_{i=\frac{k}{2}}^{k} \frac{1}{e^{\tau(\frac{k}{2}-p)}+1}
\\ & > \frac{1}{2k} \sum_{i=\frac{k}{2}}^{k} \frac{1}{e^{\tau(\frac{k}{2}-\frac{k}{2})}+1}
\\ & = \frac{1}{2k} \sum_{i=\frac{k}{2}}^{k} \frac{1}{2} = \sum_{i=\frac{k}{2}}^{k}\frac{1}{4k}   = \frac{1}{8}
\end{align*}

Thus, as the learner can get at least $\frac{1}{8}$, their best response price will always be at least $\frac{1}{8}$, as otherwise their payoff would be strictly smaller than $\frac{1}{8}$. Therefore, if the optimizer plays a uniform distribution, the follower's best response is to price at some value  $\geq \frac{1}{8}$. Thus, the payoff of the optimizer is 

\begin{align*}
u_{o}(r,BR_{l}(r)) \cdot T - \ o(T) 
    & \geq u_{o}(r,\frac{1}{8}) \cdot T - o(T) \tag{As $u_{o}(r,p)$ is increasing in $p$.} 
   \\ & \geq \mathbb{E}[r | r < \frac{1}{8}] \cdot  Pr(r < \frac{1}{8}) \cdot T - o(T)
    \\ & = \frac{1}{16} \cdot  \frac{1}{8} \cdot T - o(T)
    \\ & = \frac{1}{128} T - o(T)
    \\ & = \Omega_{k}(1) \cdot T
\end{align*}

\end{proof}

Next, we show that the benefit of these high prices also extends to the learner in the Logit model. This proof again holds regardless of $\tau$.

\begin{lemma} For any no-regret learner algorithm $\cA^{noreg}_{l}$, and for any optimizer algorithm $\cA_{o}$ in a multinomial-logit repated game, if $$U_{o}(\cA^{noreg}_{l}, \cA_{o}) \geq \Omega_{k}(T)$$ then $$U_{l}(\cA^{noreg}_{l},\cA_{o}) \geq \Omega_{k}(T)$$\label{lem:logit_equal_looting}
\end{lemma}

\begin{proof}
Say that the optimizer's payoff is $c \cdot T$ for some $c$ betwen $0$ and $1$. 

Note that, in order to achieve $c$ payoff on average, the optimizer must have priced at or above $\frac{c}{2}$ with frequency at least $\frac{c}{2}$. Assume for contradiction that the optimizer does not do this. Then the optimizer's average payoff is at most 
\begin{align*}
Pr\left[o_{p} \geq \frac{c}{2}\right] \cdot 1 + \left(1 - Pr\left [o_{p} \geq \frac{c}{2}\right ]\right) \cdot \frac{c}{2} 
& < \frac{c}{2} + \left (1-\frac{c}{2}\right ) \cdot \frac{c}{2} \\
& = c - \frac{c^{2}}{4}
\end{align*}
This is strictly less than $c$, and thus by contradiction, the optimizer must have priced at or above $\frac{c}{2}$ with frequency at least $\frac{c}{2}$. 

Given this, one fixed action that the learner could have taken is to price at $\frac{c}{2}$. The learner's average payoff would therefore be at least

\begin{align*}
& \frac{e^{\tau \frac{c}{2}}}{e^{\tau \frac{c}{2}} + e^{\tau \frac{c}{2}}}Pr \left [o_{p} = \frac{c}{2}\right] \cdot \frac{c}{2} + \sum_{p = \frac{c}{2} +\frac{1}{k}; p = p + \frac{1}{k}}^{1}\frac{e^{\tau o_p}}{e^{\tau l_p} + e^{\tau \frac{c}{2}}}Pr \left [o_{p} > \frac{c}{2}\right] \cdot \frac{c}{2}
\\& = \frac{1}{2} \cdot Pr \left [o_{p} = \frac{c}{2}\right] \cdot \frac{c}{2} + \sum_{p = \frac{c}{2} +\frac{1}{k}}^{1}\frac{e^{\tau o_p}}{e^{\tau l_p} + e^{\tau \frac{c}{2}}}Pr \left [o_{p} > \frac{c}{2}\right] \cdot \frac{c}{2} \\
& \geq \frac{1}{2} \cdot Pr \left [o_{p} = \frac{c}{2}\right] \cdot \frac{c}{2} + \sum_{p = \frac{c}{2} +\frac{1}{k}}^{1}\frac{1}{2}Pr \left [o_{p} > \frac{c}{2}\right] \cdot \frac{c}{2} \\
& = \frac{1}{2}Pr \left [o_{p} \geq \frac{c}{2}\right] \cdot \frac{c}{2} \\
& \geq \frac{c}{4} \cdot \frac{c}{2}\\
& = \frac{c^{2}}{8} 
\end{align*}

By the no-regret guarantee of the learner, their total payoff must be at least $\frac{c^{2}T}{8} - o(T) = \Omega_{k}(T)$. 

\end{proof}

We can now combine these results, to show that, if the learner deploys any no-regret algorithm in the Multinomial-logit-based model and the optimizer responds via any strategy which gets them payoff at least that of static random pricing, then prices are supra-competitive and both the learner and the optimizer get a constant fraction of the profits.

\begin{theorem} \label{thm:logit_generalized} In a Multinomial-logit-based repeated pricing game, for any no-regret learner algorithm $\cA_{l}^{noreg}$ and any optimizer algorithm $\cA_{o}$ such that 
\begin{equation}
    U_{o}(\cA_{l}^{noreg}, \cA_{o}) \geq U_{o}(\cA_{l}^{noreg}, \cA_{o}^{r})
\end{equation}
Where $\cA_{o}^{r}$ is the static uniformly random algorithm, 
\begin{equation}
    U_{o}(\cA_{l}^{noreg}, \cA_{o}) = \Omega_{k}(T)
\end{equation}
 
\begin{equation}
    U_{l}(\cA_{l}^{noreg}, \cA_{o}) = \Omega_{k}(T)
\end{equation}

And the average price is $\Omega_{k}(1)$.
\end{theorem}
\begin{proof}
    By Lemma~\ref{lem:logit_random}, the optimizer gets payoff at least $\Omega_{k}(1)\cdot T$ against any no-regret algorithm by playing static random pricing. Therefore, the optimizer's payoff utilizing any strategy which is better than static random pricing also gives them a payoff of at least $\Omega_{k}(1)\cdot T$.
    Furthermore, by Lemma~\ref{lem:logit_equal_looting}, this implies that the learner also gets a payoff of at least $\Omega_{k}(1)\cdot T$. Finally, by Observation~\ref{obs:value_transfer}, this implies that the average price is $\Omega_{k}(1)$.
\end{proof}

\end{document}